\documentclass[abstract=on,notitlepage,superscriptaddress,nofootinbib,aps,showpacs,twocolumn,prd]{revtex4}
\usepackage{srcltx}
\usepackage{epsf}
\usepackage{amsfonts}
\usepackage{amsmath}
\usepackage{mathrsfs}
\usepackage{epsfig}
\usepackage{enumerate}
\usepackage{lipsum}
\usepackage{graphicx}
\usepackage{epsf}
\usepackage{subfig}
\usepackage{color}
\usepackage{caption}
\usepackage{subfig}
\usepackage{epstopdf}
\usepackage{float}
\usepackage{dcolumn}
\usepackage{hyperref}
\usepackage{amsthm}

\newtheorem{prop}{Proposition}

\newtheorem{thm}{Theorem}

\newcommand{\be}{\begin{equation}}
\newcommand{\ee}{\end{equation}} 
\newcommand{\eei}{\end{equation}\indent\indent}
\newcommand{\bc}{\begin{center}}
\newcommand{\ec}{\end{center}}
\newcommand{\ber}{\begin{eqnarray*}}
\newcommand{\ear}{\end{eqnarray*}}
\newcommand{\ba}{\begin{array}}
\newcommand{\ea}{\end{array}}

\newcommand{\na}{\nabla}

\newcommand{\vv}{{\cal V}}

\newcommand{\bea}{\begin{eqnarray}}
\newcommand{\eea}{\end{eqnarray}}
\newcommand{\nn}{\nonumber}

\newcommand{\ei}{\end{itemize}}

\newcommand{\bra}[1]{\left(#1\right)}
\newcommand{\bras}[1]{\left[#1\right]}

\newcommand{\Lietwo}{{\cal L}}

\newcommand{\reff}[1]{(\ref{#1})}

\def\case#1/#2{\textstyle\frac{#1}{#2} }

\def\fp{f_{,R}}
\def\fpp{f_{,RR}}

 \def\prd{Phys.\ Rev.\ D }

 \def\plb{{\em Phys. Lett.\/} {\bf B}}

\def\cqg{{\em Class. Quantum Grav.\/} }

\def\aph{{\em Ann. Phys. (NY)\/} }

 {\begin{list}{}%
         {\setlength{\leftmargin}{#1} %
         \setlength{\rightmargin}{#2}%
         }%
         \item[\textit{}]%
 }
 {\end{list}}

\makeatletter

\begin{document}
\sloppy

\title{ Collapsing spherical stars in f(R) gravity}

\author{Rituparno Goswami}
\email{Goswami@ukzn.ac.za}
\affiliation{Astrophysics \& Cosmology Research Unit,
School of Mathematics Statistics and Computer Science,
University of KwaZulu-Natal,
Private Bag X54001, Durban 4000, South Africa.}

\author{Anne Marie Nzioki}
\email{anne.nzioki@gmail.com}
\affiliation{Astrophysics \& Cosmology Research Unit,
School of Mathematics Statistics and Computer Science,
University of KwaZulu-Natal,
Private Bag X54001, Durban 4000, South Africa.}

\author{Sunil. D. Maharaj}
\email{Maharaj@ukzn.ac.za}
\affiliation{Astrophysics \& Cosmology Research Unit,
School of Mathematics Statistics and Computer Science,
University of KwaZulu-Natal,
Private Bag X54001, Durban 4000, South Africa.}

\author{Sushant G. Ghosh}
\email{sgghosh@gmail.com}
\affiliation{Astrophysics \& Cosmology Research Unit,
School of Mathematics Statistics and Computer Science,
University of KwaZulu-Natal,
Private Bag X54001, Durban 4000, South Africa.}
\affiliation{Centre for Theoretical Physics,
Jamia Millia Islamia, New Delhi 110025, India.}

\begin{abstract}
We perform a careful investigation of the problem of physically realistic gravitational collapse of massive stars in $f(R)$-gravity. We show that the extra matching conditions that arise in the modified gravity impose strong constraints on the stellar structure and thermodynamic properties. In our opinion these constraints are unphysical. We prove that no homogeneous stars with non-constant Ricci scalar can be matched smoothly with a static exterior for any nonlinear function $f(R)$. Therefore, these extra constraints make classes of physically realistic collapse scenarios in general relativity, non-admissible in these theories. We also find an exact solution for an inhomogeneous collapsing star in the Starobinski model that obeys all the energy and matching conditions. However, we argue that such solutions are fine-tuned and unstable to matter perturbations. Possible consequences on black hole physics and the cosmic censorship conjecture are also discussed.
\end{abstract}
\maketitle
\nopagebreak

\section{Introduction}
In spite of the success of general relativity (GR), both in astrophysics and cosmology, alternative theories of gravity exist. In addition to theoretical considerations, such theories are motivated by the ambiguous nature of dark energy in cosmology which is responsible for the observed late time accelerated expansion of the universe. The alternative possibility, in an attempt to do away with the need for dark energy, is to conjecture that GR is an ``{\em effective}'' theory of a more general theory on cosmological scales. Among the modified theories of gravity that provide a late time acceleration for the universe, without the need for the presence of any exotic fluids, is $f(R)$-gravity. This theory is based on a gravitational action that contains an arbitrary but well defined function of the Ricci scalar $R$ \cite{DEfR,fr,fr2,fr1,scm}. Some of these models naturally admit a phase of accelerated expansion both in the early universe as an inflationary phase \cite{star80}, and also in a late time phase after passing through a matter dominated decelerating expansion \cite{shosho}. These theories essentially contain an additional scalar degree of freedom which is commonly interpreted as a scalar field called the scalaron, with the scalaron potential constructed from the Ricci scalar. An interesting aspect of this class of theories (contrary to other models that have the square of Ricci or Riemann tensor in the action) is that the  Ostrogradski instability is not problematic in these theories, despite the equations of motion being fourth-order in the metric components.

Although $f(R)$-gravity has been quite successful in providing a geometrical origin of the dark sector of the universe, it poses considerable problems in the astrophysical sector which we list below:
\begin{enumerate}[(a) ]
\item It is extremely difficult to find exact solutions of static or dynamic stellar objects as the field equations are fourth-order differential equations in the metric components.  Nevertheless, significant attention has been dedicated in finding spherically symmetric solutions in $f(R)$-theories \cite{ssfr,amn}, including the collapsing solutions \cite{gc2,alvaro,Dong}.
\item The post-Newtonian and parametrised post-Newtonian expansion of these theories put a strong constraint on the parameters of the theory which may not always be compatible with cosmological predictions. Hence the parameters of the function $f(R)$ have to be constrained by both astrophysical and cosmological observations to get a viable model.
\item The fourth-order field equations generate extra matching conditions between two spacetimes beyond the usual Israel-Darmois \cite{Israel, Darmois} conditions in GR. The extra conditions arising from the matching of the Ricci scalar and it's normal derivative across the matching surface, heavily constrict the set of useful astrophysical solutions. For any stellar object, the spacetime of the interior of the star has to be suitably matched with the exterior spacetime. Even the bottom-up picture of the universe is made up of spherical or almost spherical stellar objects immersed in vacuum, and this requires proper spacetime matching. It was recently shown that a Einstein-Strauss-like construction is not possible in non-linear $f(R)$-theories \cite{Clifton}, and these extra matching conditions lead to fine-tuning problems for static star models \cite{3GR}.
\end{enumerate}

In this paper, based on the existence of extra matching conditions, we derive a general result for all $f(R)$-theories with a non-linear function $f$: homogeneous dynamic stars with non-constant Ricci scalar cannot be matched to a static exterior spacetime. This result, though very interesting, is also quite heartbreaking in the sense that a theory of gravity should not determine the structure and thermodynamics of the star, this should be determined by stellar physics. Furthermore all homogeneous collapsing stellar models of GR no longer remain viable models in these theories. This has serious consequences on black hole physics as the most important example of black hole formation in GR is the Oppenheimer-Snyder-Datt \cite{osd} collapse which is the model of a collapsing homogeneous dust ball with an exterior Schwarzschild spacetime. Such models are no longer admissible in these $f(R)$-theories, and hence we need to find other examples of physically realistic matter collapse leading to a black hole. The detailed investigations on the gravitational collapse of homogeneous matter in modified gravity \cite{alvaro} also become redundant in this case.

Now the key question is: Is it possible to find an exact or numerical solution of a collapsing stellar model in $f(R)$-gravity with physically realistic matter that satisfies all the matching conditions with a Schwarzschild exterior? Since this is not possible with a homogeneous star, in this paper we find an exact inhomogeneous solution in the Starobinski model with $f(R)=R+\alpha R^2$, where the collapsing stellar matter has anisotropic pressure and heat flux. It is quite interesting to observe that this solution mimics the Lemaitre-Tolman-Bondi dust solution \cite{ltb,tb,hb} in GR. Thus we show that the space of such physically interesting solutions where the collapsing star can be matched to a static Schwarzschild exterior, is non-empty. Also, in spite of standard matter obeying all physically realistic energy conditions, we argue that to find such solutions we need a considerable fine-tuning of the thermodynamic properties of the collapsing star. Hence these solutions seem to be unstable with respect to the matter perturbations.

The paper is organised as follows: In the next section we give a brief overview of the field equations in $f(R)$-gravity, energy momentum tensor for the collapsing stellar matter and the energy conditions. In section 3, we discuss the matching conditions in these theories and state the no-go proposition for homogeneous stars. In section 4, we present an exact collapsing solution of a spherically symmetric star in the Starobinski model with $f(R)=R+\alpha R^2$, that obeys all the required matching conditions with the vacuum, static Schwarzschild exterior. In the final section we discuss the stability and genericity of such solutions.

We use units which fix the speed of light and the
gravitational constant via $8\pi G = c^4 = 1$, and throughout this paper we used the metric signature $+2$.

\section{Field equations and Energy conditions}
In order to study spherically symmetric solutions in $f(R)$-gravity, we begin by modifying the Einstein-Hilbert action. 
The modification in $f(R)$-gravity, is obtained by generalising the Lagrangian in the Einstein-Hilbert action 
such that the Ricci scalar $R$ is replaced with a function $f(R)$, so that
\be
\label{action}
{\cal S}= \frac12 \int dV \bras{\sqrt{-g}\,f(R)+ 2\,\Lietwo_{M}(g_{ab}, \psi) } ~.
\ee
where $ \Lietwo_{M} $ is the Lagrangian density of the matter fields $\psi$,  $g$ is
the determinant of the metric tensor $g_{ab}$ $(a,b=0,
1, 2, 3)$, $R$ is the scalar curvature and $f(R)$ is the real
function defining the theory under consideration. Varying the action \reff{action} with respect to the metric over a 4-volume
yields
\be
\label{field}
G_{ab} \, \fp -  \frac12 g_{ab}\,(f-R \,  \fp)
 -\na_{a}\na_{b}\fp +g_{ab} \, \Box  \fp =   T^{M}_{ab}~,
\ee
where $\fp =d f(R)/dR $, $\Box \equiv \na_{c}\na^{c}$, $G_{ab}$ is the usual Einstein tensor and
$T^{M}_{ab} $ is the matter energy momentum tensor (EMT) defined by \cite{Herrera}
\be
 \label{EMT}
T^{M}_{ab}=- \frac{2}{\sqrt{-g}} \, \frac{\delta \Lietwo_{M}} {\delta g^{ab} } ~.
\ee
It can be seen that for the special case $f(R) = R$, the field equations \reff{field} reduce to the
standard Einstein field equations. These
theories are also  known as fourth-order gravity, since the
term $( g_{ab} \Box - \nabla_a \nabla_b )\fp $ has fourth-order
derivatives with respect to the metric.  

We now consider a spherically symmetric spacetime whose geometry is
determined by the metric
\be 
\label{genmetric} 
ds^2 = - e^{2\nu
(t,r)} dt^2 + e^{2\psi (t,r)} \,dr^2 + C^2(r,t) \, d\Omega^2 ~, 
\ee
where $ d\Omega^2  = d\theta^2 +\sin^2\theta \, d \phi^2 $. In
terms of its components, the EMT is defined as 
\be 
\label{decompEMT}
T_{ab} = \mu\, u_{a} \,u_{b} + p\,h_{ab} + 2\, q u_{(a}\, n_{b)} -
\Pi \bra{ n_{a} \,n_{b}-\frac13 h_{ab}}~, 
\ee 
where $\mu$ is the
energy density,  $u^{a}$ is the four-velocity of the fluid
satisfying $u_{a}u^{a} = -1$, the pressure term $p$ is given by 
\be
p= \frac{p_{r} + 2 \,p_{t}}{3} \;,
\ee 
where $p_{r}$ and $p_{t}$ are,
respectively the radial and tangential components of the pressure,
$h_{ab}$ is the spatial projection tensor defined as 
\be 
h_{ab} =u_{a} u_{b} + g_{ab}~,
 \ee 
 $n^{a}$ is
the spatial unit vector in the radial direction with the properties
$n_{a}n^{a} = 1, ~n_{a}u^{a} = 0$, the term $q = q_{a} n^{a} $ 
is the component along $n^{a}$ of radial energy flux vector $q^{a}$, and $ \Pi =
p_{t} - p_{r}$ measures the amount of anisotropy ($p_{r} = p_{t}$ corresponds to isotropic
pressure and $p_{r} \neq p_{t}$ for anisotropic pressure). 
We wish to find the general solution of the modified Einstein field
equations (\ref{field}) with the metric (\ref{genmetric}), which contains three
arbitrary functions.  

In order for the matter field to be physically realistic it must obey
one or all of the energy conditions \cite{HE}. For a spherically symmetric fluid with the EMT
\reff{decompEMT}, for the energy conditions 
to be satisfied we must satisfy the following inequalities \cite{Kolassis, Chan}: 
\bea
\label{EC1}
&\lvert \mu + p_{r} \rvert- 2\, \lvert q\rvert \geq 0 ~,\\
\label{EC2}
&\mu - p_{r} + 2\,p_{t}+\bigtriangleup \geq 0 ~,
\eea
 together with
{\begin{enumerate}[(a)]
 \item{\textit{weak} energy conditions (WEC)}
\be
\label{WEC}
\mu - p_{r} +\bigtriangleup \geq 0 ~,
\ee
\item{\textit{dominant} energy conditions (DEC)}
\bea
\label{DEC1}
&\mu - p_{r} \geq 0 ~, \\
\label{DEC2}
 &\mu - p_{r} -2\,p_{t} +\bigtriangleup \geq 0 ~,
\eea
\item{\textit{strong} energy conditions (SEC)}
\be
\label{SEC}
2\,p_{t}+ \bigtriangleup \geq 0~,
\ee
\end{enumerate}}
where $\bigtriangleup = \sqrt{\bra{\mu + p_{r}}^2 - 4\,q^2}$.

We can write the components of the $G_{ab}$ tensor, for which the
metric \reff{genmetric} satisfies the field equations \reff{field}, as follows:
\bea
\label{genGab1}
G^{0}{}_{0} &=& -\,\frac{1}{ C^2 }
\bras{1+e^{-2\nu} \bra{\dot{C}^2
+2 \,\dot{\psi} \,C\,\dot{C} } \right. \nn \\  
&&	\left.
-\,e^{-2 \psi} \bra{{C^{\prime}}^2
+2\,C\, C^{\prime\prime}-2\,\psi^{\prime} \,C\, C^{\prime}}}  ~,
\\
\label{genGab2}
 G^{1}{}_{1} &=& -\,\frac{1} { C^2 }
\bras{1- e^{-2 \psi} \bra{{C^{\prime}}^2
+2\,\nu^{\prime} \,C\, C^{\prime}} \right. \nn \\  
&&	\left.
+ \,e^{-2\nu} \bra{\dot{C}^2
+2\,C\, \ddot{C}-2 \,\dot{\nu}  \,C\,\dot{C}} }  ~,
\\
\label{genGab3}
 G^{2}{}_{2} &=& G^{3}{}_{3} =
-\,\frac {1}{ C}
\bras{e^{-2\nu} \bra{\dot{\psi} \, \dot{C}- \dot{\nu}\,\dot{C}
+\dot{\psi}^{2} \, C+\ddot{C} \right. \right. \nn\\
&&\left.\left. +\, \ddot{\psi} \, C 
- \dot{\nu} \, \dot{\psi} \,  C }  + e^{-2\psi} \bra{ \psi^{\prime} \, C^{\prime}
- \nu^{\prime} \,C^{\prime} - {\nu^{\prime}}^{2}  \,C \right. \right.\nn\\
&&\left. \left.
 -\,C^{\prime\prime}
 - \nu^{\prime\prime}\, C   + \nu^{\prime} \, \psi^{\prime} \,C}}~,
 \\
 \label{genGab4}
G^{0}{}_{1} &=&  \frac{2}{  e^{2\nu} \,C  }\bras{\dot{C}^{\prime}
- \nu^{\prime} \,\dot{C}- \dot{\psi}\,C^{\prime}}~,
\eea
where ( $ \dot{}$ ) and (${}^\prime$ ) denote the partial derivative
with respect to $t$ and $r$ respectively. The Ricci scalar for the
metric is
\begin{widetext}
\bea
\label{genRicci}
R&=& \frac {2}{C^{2}}\bras{1
-e^{-2\psi} \bra{2\,\nu^{\prime}\,  C \, C^{\prime}
- \nu^{\prime} \,\psi^{\prime}\,  C^{2}+ \nu^{\prime\prime}\, C^{2}
+ {\nu^{\prime}}^{2}\,  C^{2}  -2\, \psi^{\prime}\, C \,C^{\prime}
+2\, C \, C^{\prime\prime} +{C^{\prime}}^{2}} \right. \nn\\
&&\left.\qquad  \qquad-e^{-2\nu} \bra{2\,\dot{\nu}\,C\, \dot{C}
+\dot{\nu} \, \dot{\psi}\, C^{2} - \ddot{\psi} \, C^{2}- \dot{\psi}^{2}\,C^{2}
-2\,\, \dot{\psi}\,C\, \dot{C}-2 \,C \,\ddot{C} -\dot{C}^{2}} }\;.
\eea
\end{widetext}

\section{Junction conditions: A no-go proposition}
In order to study gravitational collapse,
it is necessary to describe adequately the geometry of the interior
and exterior regions and to give the conditions which allow matching of these regions.
Any astrophysical object is immersed in vacuum or almost vacuum spacetime 
(like any star within the stellar system), and hence the exterior spacetime around 
a spherically symmetric star is well described by the Schwarzschild geometry. 
Therefore any physically realistic star should be matched with a static vacuum 
solution which in the case of spherical symmetry is
the Schwarzschild geometry in GR.

We consider matching two spacetimes $\vv^{\pm}$ with the boundary surface
denoted by $\Sigma$. The junction surface must be the same in $\vv^{+}$ and
$\vv^{-}$, which implies continuity of both the metric and the extrinsic curvature
of $\Sigma$ as in GR \cite{Israel, Darmois}. Moreover, in $f(R)$-theories of gravity, continuity of the Ricci scalar across the boundary surface and 
continuity of its normal derivative are also required \cite{Deruelle, Clifton, Senovilla}. \\

To understand the above in some detail let us write the metric of the interior and exterior 
spacetime locally (near the matching surface) in terms of the Gaussian coordinates
\be
\label{gaussmetric}
ds^2 =g_{ab} \,d\xi^{a}\,d\xi^{b} = d\tau^2 + \gamma_{ij} \, d\xi^{i}\,d\xi^{j} ~,
\ee
where $\xi^{i}, ~i = 1,2,3$ are the intrinsic coordinates to $\Sigma$,
$\gamma_{ij}$ is the intrinsic metric (first fundamental form) of $\Sigma$
and the boundary is located at $\tau =0$. Given \reff{gaussmetric}, together
with the extrinsic curvature (second fundamental form) of the boundary
surface defined by
\be
K_{ij} = -\frac{1}{2}\,\partial_\tau \gamma_{ij}~,
\label{exrinsic}
\ee
the Ricci scalar can be written as
\be
\label{R}
R = 2\, \partial_\tau K - \tilde{K}_{ij}\,\tilde{K}^{ij} - \frac{4}{3} K^2 + \cal{R}~,
\ee
where $\cal{R}$ is the Ricci curvature constructed from the 3-metric
$\gamma_{ij}$, $K$ is the trace part of the extrinsic curvature and
$\tilde{K}_{ij}$ is the trace-free part.

The continuity requirements at the boundary lead to the following
junction conditions in $f(R)$-theories
\bea
\label{junction1}
\left[\gamma_{ij} \right]^+_- = 0~,\\
\label{junction2}
\fpp \left[ \partial_\tau R \right]^+_-  = 0~,\\
\label{junction3}
\fp  \left[ \tilde{K}_{ij} \right]^+_- = 0\;,\\
\label{junction4}
\left[ K \right]^+_- = 0~,\\
\label{junction5}
\left[ R \right]^+_- = 0~,
\eea
provided $\fpp \neq 0$. For further details, we refer the reader
to \cite{Deruelle, Senovilla}. It is worth noting that the conditions 
(\ref{junction2}) and (\ref{junction5}) are the extra conditions that 
arise in $f(R)$-theories with non-linear function $f$. These 
extra conditions are necessary for the continuity of the field equations 
across the matching surface, and indeed impose some considerable
 constraints on viable spacetimes which we describe below in 
 the proposition.

\begin{prop}
For any $f(R)$-gravity with a non-linear function $f$, a dynamic homogeneous spacetime with non-constant Ricci scalar cannot be matched with a static spacetime across a fixed boundary.
\end{prop}
\begin{proof} 
If the dynamic homogeneous spacetime has non-constant Ricci scalar then the Ricci scalar will evolve with time on one side of the boundary, whereas on the other side of the boundary the Ricci scalar remains constant as the spacetime is static. Hence the junction condition (\ref{junction5}) can never be satisfied for all epochs.
\end{proof}

The above proposition immediately nullifies the existence of homogeneous dynamic stars with non-constant Ricci scalar for example, collapsing dustlike matter as in Oppenheimer-Snyder-Datt model \cite{osd} which is smoothly matched to the exterior Schwarzschild spacetime, in these theories. On the same note the Einstein-Strauss construction is not possible. The only homogeneous collapsing stars that can be matched to a static exterior are those which have a constant Ricci scalar in the interior (examples are Vaidya, dS/AdS Vaidya, charged Vaidya or collapsing perfect fluids with a null equation of state). Thus the modification in the theory of gravity heavily constrains and dictates the structure and the thermodynamic properties of the collapsing star. This is surely unphysical as the stellar properties should be the outcome of the stellar physics as the gravitational collapse of the star commences. 

Another important question that may arise here is: what happens if the exterior is non-static? The solar system experiments constrain heavily such a scenario, and the time variation must be in the cosmological time scale. However, the important point is, the time scale of gravitational collapse of  massive stars is much smaller than the cosmological time scales. Hence if the exterior is non-static, the matching of the Ricci scalar and the normal derivative at the surface of a homogeneous star is  still not possible. One possible way to avoid such a scenario arises if we allow a ``jump" in the curvature terms in the field equations. In other words we do not match both the Ricci scalar and itÕs normal derivative simultaneously. This will result in surface stress energy terms, that are purely generated by the dynamic curvature. However such surface stresses on realistic collapsing stars must have observational signatures and can be established via experimental evidences.

Furthermore, we have to find a suitable model explaining the dynamical black hole formations in these theories. Although a toy model, the Oppenheimer-Snyder-Datt collapse is widely believed to be a general model of black hole formation \cite{rp}. The geometry of the trapped surfaces changes considerably when the collapsing matter is inhomogeneous (like in the case of Lemaitre-Toman-Bondi collapse in GR \cite{psj,dc}) and in many other cases where a locally naked central singularity develops (as in the solution of the next section). Existence of a Cauchy horizon due to a naked singularity can prevent the spacetime to be future asymptotically simple, and hence the general global proofs in most of the theorems of black hole dynamics and thermodynamics have to be reanalysed \cite{HE}.

\section{An exact collapsing solution in $f(R) = R+ \alpha\,R^2$ gravity} 

Having established that no homogeneous collapsing star can be matched to a static exterior, the key question here is, whether it is possible to find a physically realistic inhomogeneous collapsing stellar solution in $f(R)$-gravity. In this section we address this question in the case of $f(R)=R+\alpha\,R^2$. By the phrase ``physically realistic'' we mean the following:
\begin{itemize}
\item The collapsing stellar matter should obey all the energy conditions.
\item At the comoving boundary of the collapsing star, the interior spacetime should matched smoothly with a Schwarzschild spacetime as all experimental tests in the solar system indicate that the spacetime outside the Sun is well described by the Schwarzschild geometry.
\end{itemize}

In order to satisfy the second condition above, we must choose a non-linear function $f$ that has Schwarzschild spacetime as a vacuum solution. For this, we recall the extension of Birkhoff's theorem for $f(R)$-gravity \cite{amn} which states that 
\begin{thm}
For all functions $f(R)$ which are of class $C^3$ at $R=0$ and
$f(0)=0$ while
$f'(0)\ne 0$, the Schwarzschild solution is the only static spherically symmetric vacuum 
solution with vanishing Ricci scalar.
\end{thm}
As the simplest higher order extension to GR which satisfies all the requirements of the theorem above, we consider the gravitational action proposed by Starobinski \cite{star80}, where $f(R)=R+\alpha\,R^2$. This model naturally produces a early time inflationary expansion in cosmology and it is not ruled out by the recent Plank data. At late times, this model gives a geometrical origin of dark energy and it can be shown that the cosmological solution tends to a deSitter solution in the far future and thus mimics the $\Lambda$-CDM cosmology without the fine-tuning problem of the cosmological constant $\Lambda$. In what follows, we give an exact solution for a collapsing star, with anisotropic pressure and heat flux in the interior, in $f(R) = R+ \alpha\,R^2$ theory. The matter in the interior of the cloud is described by the following EMT distribution
\reff{decompEMT} with components
\begin{widetext}
\bea
\label{energy}
\mu &=& -\frac{1}{2 {C^{\prime}}^{5} C^{4}}
\bras{  4\, \alpha \,C^{2} 
\bra{ F^{\prime \prime \prime} {C^{\prime}}^{2}
+2\,F^{\prime} C C^{\prime \prime }  
+F^{\prime} C^{2}  {\dot{C}^{\prime 2}}}
-\alpha \, {C^{\prime}}^{3}\bra{3\, {F^{\prime}}^{2} 
+8\, F^{\prime\prime} C
+4\,  {\dot{F}^{\prime}}C^{2} {\dot{C}^{\prime}} 
-16\,F^{\prime} C \dot{C} {\dot{C}^{\prime}}   }
 \right. \nn\\&&\left.
+2\,{C^{\prime}}^{4}
\bra{4 \,\alpha\,F^{\prime} 
- F^{\prime} C^{2}
-4\, \alpha\,\dot{F}^{\prime}  C \dot{C}  
+8\,\alpha \, F^{\prime}  \dot{C}^{2}}
+12\,\alpha\, F^{\prime} C^{2}{C^{\prime\prime}}^{2} 
-12\,\alpha \, F^{\prime \prime} C^{2} C^{\prime} C^{\prime \prime} 
-4\,\alpha \, F^{\prime} C^{2} C^{\prime} C^{\prime \prime \prime} 
}\;,
\eea

\bea
 \label{prad}
p_{r} &=& -\frac{\alpha}{2{ C^{\prime }}^{3} C^{4}}
\bras{- \, C^{\prime}
\bra{ 8\, F^{\prime \prime} C  
+ {F^{\prime}}^{2} C^{\prime}
+8\,  \dot{F}^{\prime} C^{2} \dot{C}^{\prime} 
-8\,  F^{\prime}  C \dot{C} \dot{C}^{\prime} 
 +4\, F^{\prime} C^{2} \ddot{C}^{\prime}   }
+4\, {C^{\prime}}^{2}\bra{
\ddot{F}^{\prime}  C^{2}
-2\, \dot{F}^{\prime} C \dot{C} 
+2\, F^{\prime} \dot{C}^{2}
 \right.\right. \nn\\&&\left.\left.
-2\, F^{\prime} C \ddot{C}  }
+16\,F^{\prime}{C^{\prime}}^{2} 
+8\,F^{\prime} C^{\prime \prime}C  
+8\, F^{\prime} {\dot{C}}^{\prime \,2} C^{2}}\;,
\eea

\bea
 \label{ptan}
p_{t} &=& -\frac{\alpha}{2{C^{\prime}}^{5} C^{4}}
\bras{ 4 \, C^{\prime }\bra{3\, F^{\prime \prime} C^{2}C^{\prime \prime} 
+ F^{\prime}  C^{2}C^{\prime \prime \prime} }
-4\, {C^{\prime}}^{2}\bra{3\, F^{\prime} C^{2}
+ F^{\prime \prime \prime}  C^{2}
+3\, F^{\prime} C C^{\prime \prime} 
- F^{\prime} C^{2}{ \dot{C}}^{\prime \,2}}
-{C^{\prime}}^{3}\bra{   {F^{\prime}}^{2}
\right. \right. \nn\\&&\left. \left.
-12 F^{\prime \prime}C
+ 4\, {\dot{F}}^{\prime}C^{2}{ \dot{C}}^{\prime}  
 -4\, F^{\prime} C \dot{C} { \dot{C}}^{\prime }
+4\, F^{\prime} C^{2}\ddot{C}^{\prime} }
 - 4 \,{C^{\prime}}^{4}\bra{
4 F^{\prime}
- \ddot{F}^{\prime}  C^{2}
 +3 \,{ \dot{F}}^{\prime } C \dot{C} 
-4\,  F^{\prime}\dot{C}^{2}+ 2 \, F^{\prime} C\ddot{C}  }
 }\;,
\eea

\bea
 \label{heat}
q &=& -\frac{2 \alpha}{ {C^{\prime}}^{4} C^{4}}
\bras{- C^{\prime}  \bra{ {\dot{F}^{\prime}}  C^{2}C^{\prime \prime} 
-2\,  F^{\prime } C C^{\prime \prime} \dot{C}
+ F^{\prime} C^{2}\dot{C}^{\prime \prime} 
 +2\,  F^{\prime \prime} C^{2}\dot{C}^{\prime}  }
{C^{\prime}}^{2} \bra{ \dot{F}^{\prime \prime}  C^{2}
 -2\, F^{\prime \prime} C \dot{C} 
+2\, F^{\prime} C \dot{C}^{\prime }  }
\right. \nn\\&&\left.
-{C^{\prime}}^{3} \bra{2\,\dot{F}^{\prime }C
-6\, F^{\prime }  \dot{C}  }
+3\, F^{\prime} C^{2} C^{\prime \prime} \dot{C}^{\prime} 
  }\;,
\eea
\end{widetext}
 where we define the function $F = F(r)$ by
\be
\label{massfn}
F\equiv r^{3}\, M(r) ~.
\ee
The quantity $M(r) $ is an arbitrary function and the function $C = C(r,t)$ has the form
\be
\label{physradius}
C = r \, \bra{1- \frac32 \sqrt{M(r)}\,t}^{\frac23} ~.
\ee
We can write \reff{physradius} as $C(t,r) = r\, a(t,r)$ with
\be
\label{scale}
a(r,t) =\bra{1- \frac32 \sqrt{M(r)}\,t}^{\frac23}\;,
\ee
representing the inhomogeneous  scale factor. We see from this that the scale
factor behaves as $a =1$ at initial epoch $t=t_i$. It is also clear that $\dot{a} < 0$
in accordance with gravitational collapse.

The metric 
\be 
\label{exactmetric} ds^2 = - dt^2 + C^{\prime{}
2}(r,t) \,dr^2 + C^2(r,t)\, d\Omega^2 ~\;,
\ee 
gives the geometry as realised in the spacetime that corresponds to the structure of the
EMT \reff{decompEMT} with components \reff{energy}-\reff{heat}. The
Ricci scalar for this metric is evaluated as
\be 
\label{Ricci1}
R=\frac{F^{\prime}}{C^{2}\,C'} ={\frac {3\, M +r\,M^{\prime} }{a^{2}
\bra{a +r \, a^{\prime}   } }}
 ~.
\ee
It is apparent that \reff{exactmetric} takes the same form as the Lemaitre-Tolman- Bondi
(LTB) dust model of general relativity.\\


In order to match a spherically symmetric collapsing cloud at the
boundary to an exterior spacetime, the junction conditions \reff{junction1}-\reff{junction5} 
have to be considered. We describe the interior spacetime
$\vv^{-}$ of the collapsing cloud by the metric \reff{exactmetric} and the
exterior spacetime $\vv^{+}$ by the Schwarzschild vacuum. These two spacetimes are matched 
at the surface of the star $\Sigma$ which is denoted by the comoving shell labelling coordinate $r=r_b$. 
Now the metric of the exterior spacetime is
\be
\label{Schwarz}
ds^2=-  \bra{1-\frac{2m}{r_s}}dt^2 + \frac{dr^2}{\bra{1-\frac{2m}{r_s}}}
+r_s^2\,d\Omega^2~\;,
\ee
where $r_s$ is the Schwarzschild radius. The junction conditions  \reff{junction3} and \reff{junction4} 
then imply that on the surface 
\be
G^1_1\vert_{\Sigma}=0\;.
\ee
Since the solution is exactly same as the LTB dust solution in GR, we can easily check that throughout 
the interior spacetime $G^1_1=0$, and hence these two matching conditions are automatically satisfied. 
The extra junction conditions \reff{junction2} and \reff{junction5} imply that
$R_{-} \vert_{\Sigma}= R^{\prime}_{-} \vert_{\Sigma}= 0$, and from these we can deduce that 
the Ricci scalar should have the form
\be
\label{Ricci2}
R \equiv \bra{r_b-r^{2}}^{2}\,g(r,a)~,
\ee
where $g(r,a)$ is a well defined and at least ${\cal{C}}^4$ function of $r$ and $a$. Then from 
\reff{Ricci1} and \reff{Ricci2} we conclude that
\be
3 \, M+ r\,M^{\prime} \equiv \bra{r_b^2-r^{2}}^{2}\,h(r)~,
\ee
where $h(r)$ is a well defined and at least ${\cal{C}}^4$ function of $r$.
The function $M(r)$ can be Taylor expanded in even powers of $r$
such that
\be
\label{mass}
M = M_{0}+r^{2}\,M_{2}+r^{4}\,M_{4}+r^{6}\,M_{6}~.
\ee
and upon substituting this expression into \reff{Ricci1}, we see that there exists values of the coefficients $M_n$ for which the junction
conditions for the Ricci scalar $R$ are fulfilled. This choice for
$M(r)$ also ensures smoothness of the initial data. \\

For our model, without any loss of generality, we may choose $r_b=1$. To transparently show that there exists values of $M_n$ for which the junction
conditions are satisfied, let us choose $h(r) = \bra{a+b\,r^{2}}$ so that
\be
\label{massval}
M_{0} = \frac{a}{3} ~, ~~ M_{2}= \frac{\bra{b-2\,a}}{5}~,
~~ M_{4}=\frac{\bra{a-2\,b}}{7}~, ~~ M_{6}=\frac{b}{9}~.
\ee
We can easily see that the above values of the coefficients will ensure that the Ricci scalar is of the form (\ref{Ricci2}). Thus the extra matching conditions impose strong constraints on the otherwise free function in GR. This means that (unlike GR) any smooth function $M(r)$ is not a physically realistic function in the Starobinski model. We can easily check that $M(r)$ is closely related to the initial density, pressures and heat flux profiles of the star. We see that only those initial data profiles that satisfy the constraints on $M(r)$ are admissible, and hence we may conclude that we need some fine-tuning of the otherwise free parameters to get a solution in $f(R)$ theories. In this regard it is expected that any form of matter perturbation in the interior of the star will disturb this fine-tuning and make the solution unstable.

To check, whether the collapsing matter obeys all the physically reasonable energy conditions with values $a= 3,\, b = -\,3$, we plot the radial profile of $M$ in
Fig \ref{massfig} and the Ricci scalar in Fig \ref{Riccifig} over the radial range 
of 0 to 1.

\begin{figure}[H]
\centering
\includegraphics[width=2.2in]{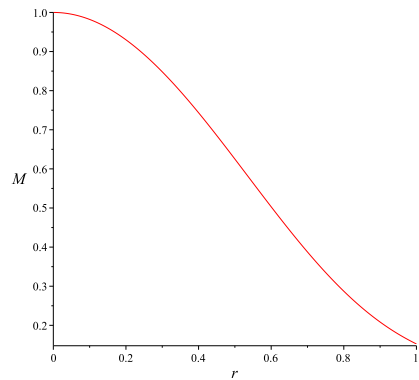}
 \caption{\small The radial profile of the function $M$.
 \label{massfig}}
\end{figure}

\begin{figure}[H]
\centering
\includegraphics[width=2.5in]{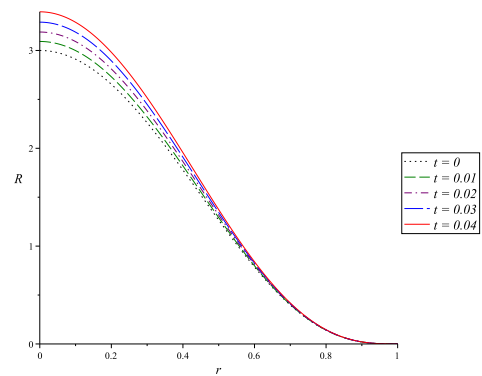}
\caption{\small The radial profile of the
Ricci scalar.
\label{Riccifig}}
\end{figure}

Taking $\alpha = 10^{-10} $, we plot the variation of these
thermodynamic terms against the radius of the 
star. We see in Fig \ref{densityfig} that the energy density is non-negative 
in the interior of the star, with a finite value at the centre and decreases 
with the radial distance with zero value at the boundary. 

\begin{figure}[H]
\centering
\includegraphics[width=2.5in]{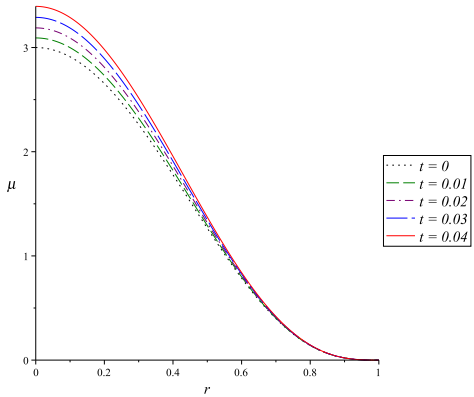}
\caption{\small The radial profile of the
energy density at different times.
\label{densityfig}}
\end{figure}

Fig \ref{Pressureradfig} shows that the radial pressure takes on 
negative values. It is finite at the star's centre and vanishes at the boundary. 
On the other hand, in Fig \ref{Pressuretanfig} the tangential pressure has a 
negative finite value at the centre which increases as a function of the radius,
becomes positive and, after reaching a maximum, decreases and vanishes at 
the boundary.

\begin{figure}[H]
\centering
\includegraphics[width=2.6in]{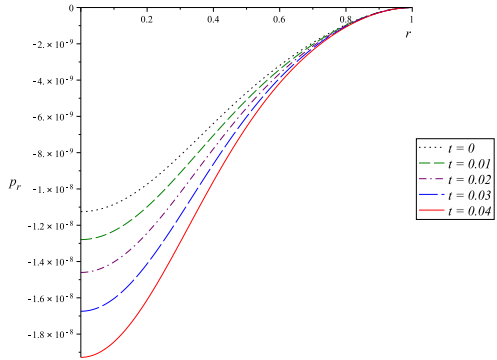}
\caption{\small The radial profile of the
radial pressure at different times.
\label{Pressureradfig}}
\end{figure}

\begin{figure}[H]
\centering
\includegraphics[width=2.6in]{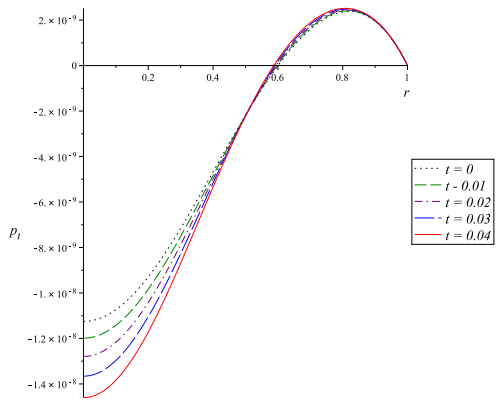}
\caption{\small The radial profile of the
tangential pressure at different times.
\label{Pressuretanfig}}
\end{figure}

The variation of pressure anisotropy $\Pi = p_{t}  - p_{r}$ with respect to the radius of the star 
is shown Fig \ref{Anisotropyfig}. The pressure anisotropy is positive in nature 
($p_{t} > p_{r}$), vanishing at both the centre and boundary. The former 
being a requirement for regularity at the centre of the star and the latter is 
required for the collapsing interior to be matched to \reff{Schwarz}.

\begin{figure}[H]
\centering
\includegraphics[width=2.6in]{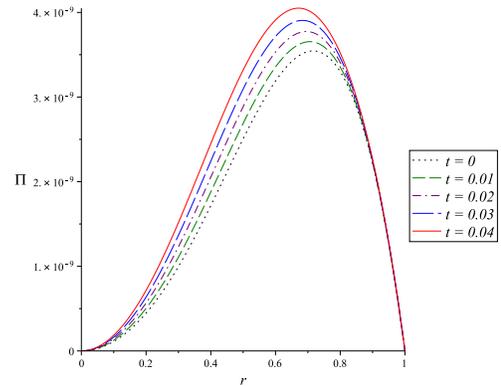}
\caption{\small The radial profile of the
anisotropy parameter at different times.
\label{Anisotropyfig}}
\end{figure}

\begin{figure}[H]
\centering
\includegraphics[width=2.6in]{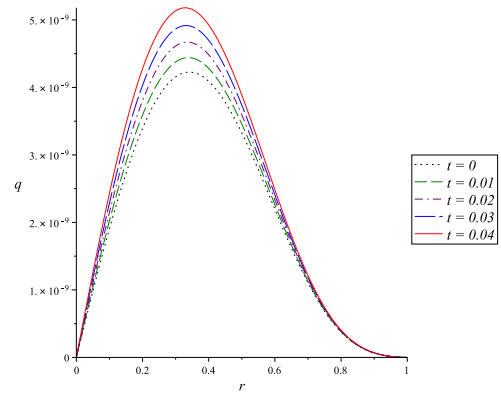}
\caption{\small The radial profile of the
heat flux at different times.
\label{Heatfluxfig}}
\end{figure}

We confirm that the energy conditions \reff{EC1}-\reff{SEC} are
satisfied in Figs \ref{ECfig}. We have considered the two inequalities of the general energy conditions (EC), the extra inequality of weak energy conditions (WEC), strong energy conditions (SEC) and the two extra inequalities of dominant energy conditions (DEC) separately. It is interesting to observe that in spite of negative pressure conditions 
the energy conditions are satisfied in this model.

\begin{figure*}[!htbp]
\centering
\subfloat[ECI \label{ECIfig}]{%
\includegraphics[scale=0.35]{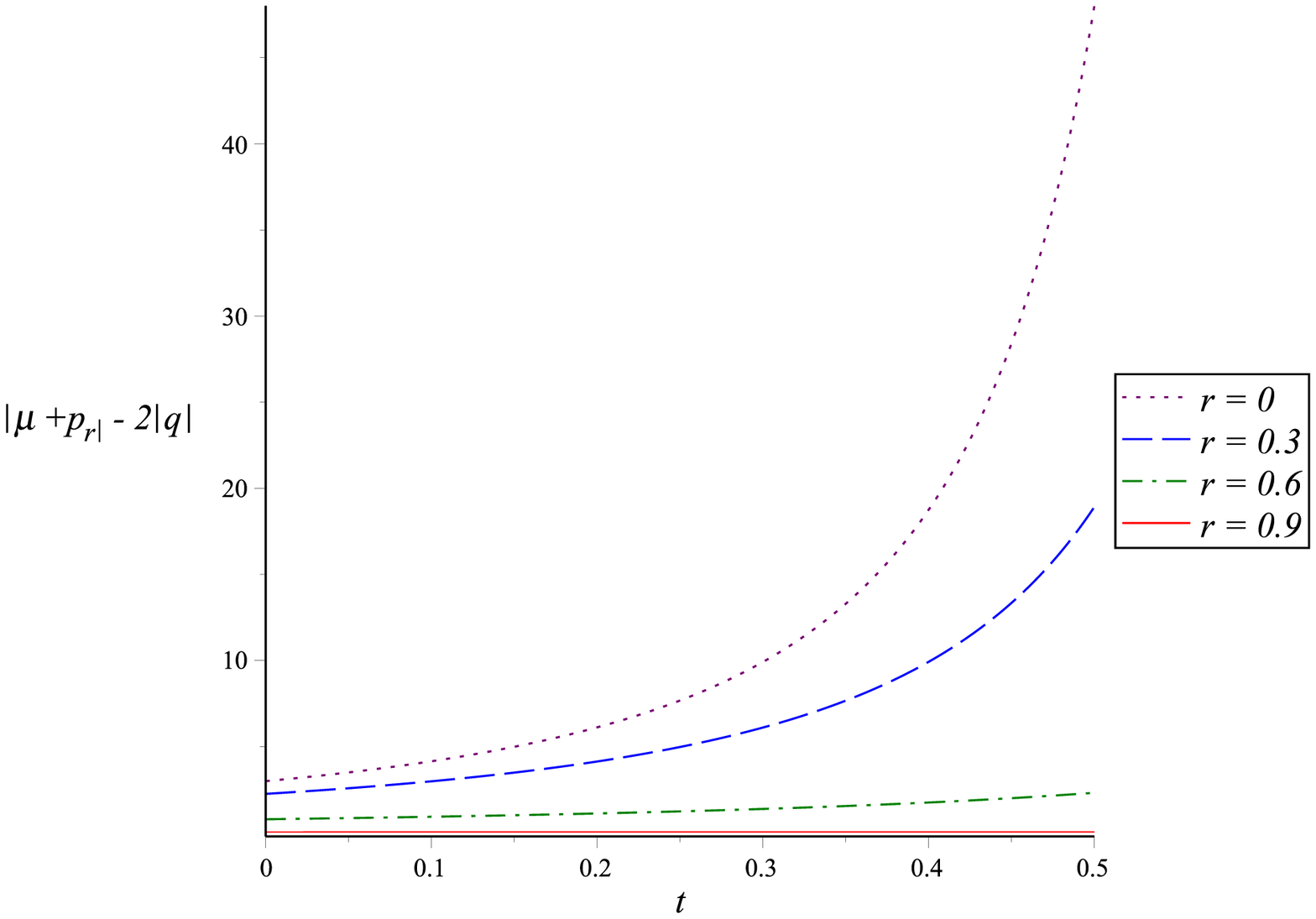}
}
\qquad \quad
\subfloat[ECII  \label{ECIIfig}]{%
\includegraphics[scale=0.35]{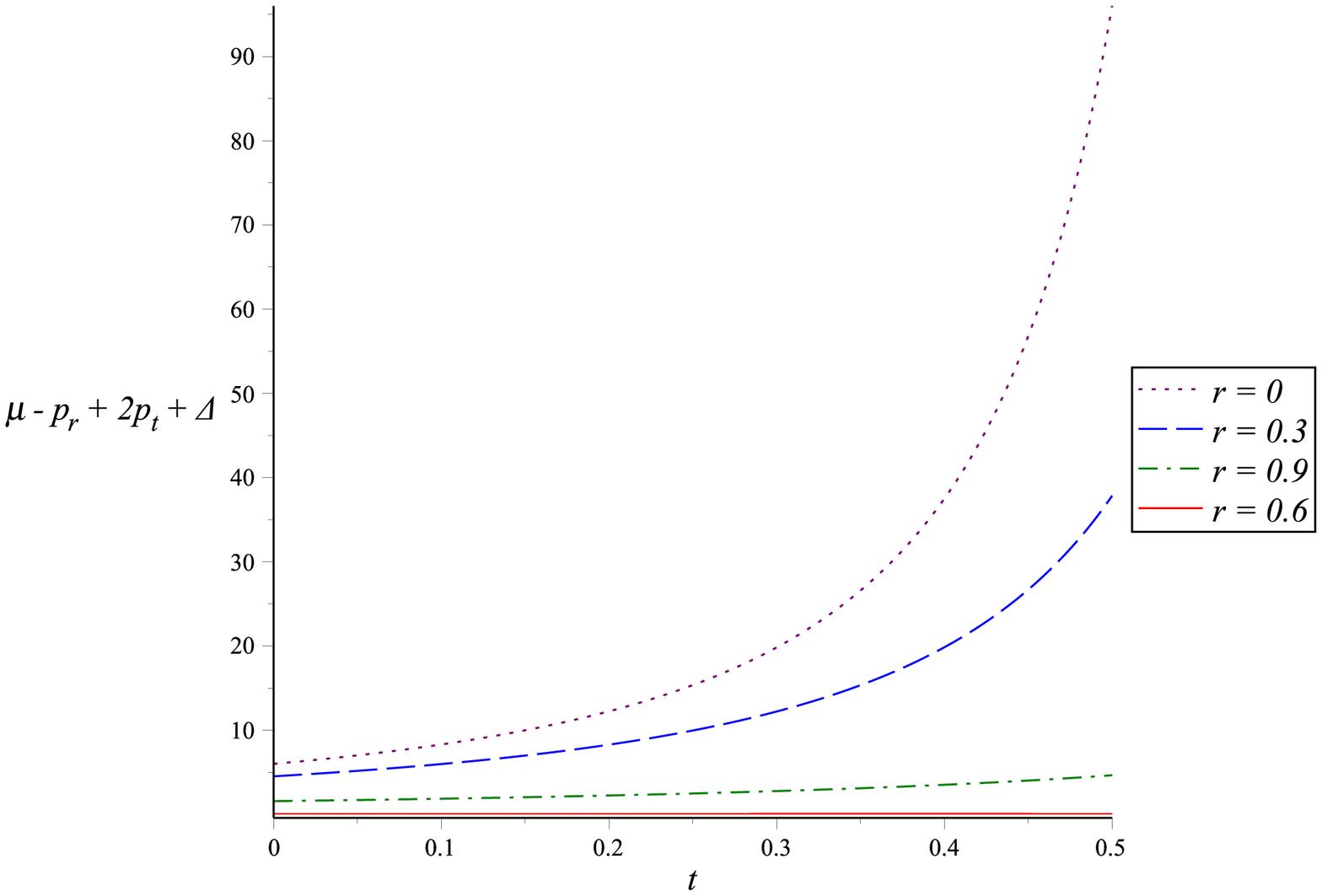}
}\\
\subfloat[WEC \label{WECfig}]{%
\includegraphics[scale=0.33]{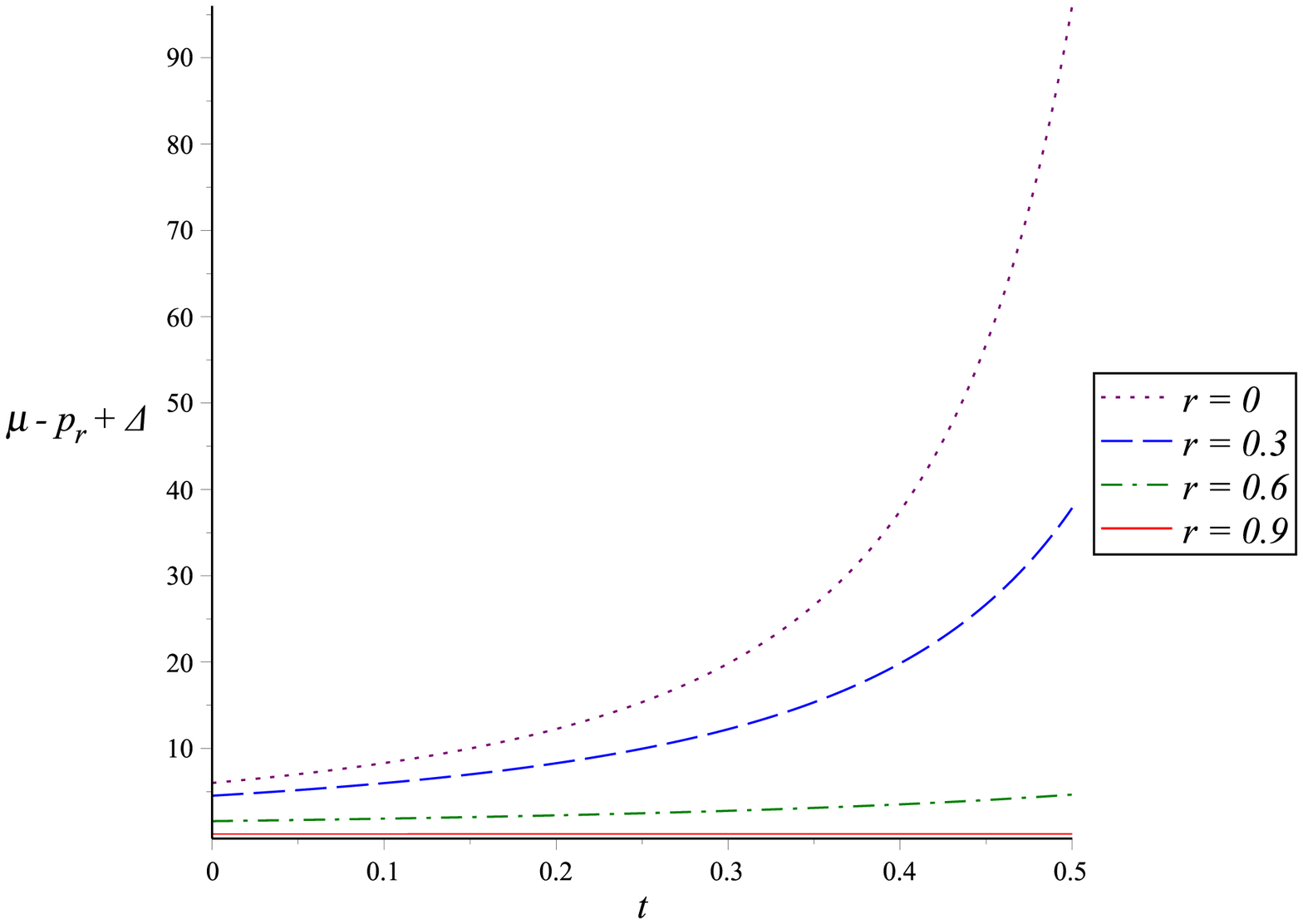}
}
\qquad \quad
\subfloat[DECI  \label{DECIfig}]{%
\includegraphics[scale=0.33]{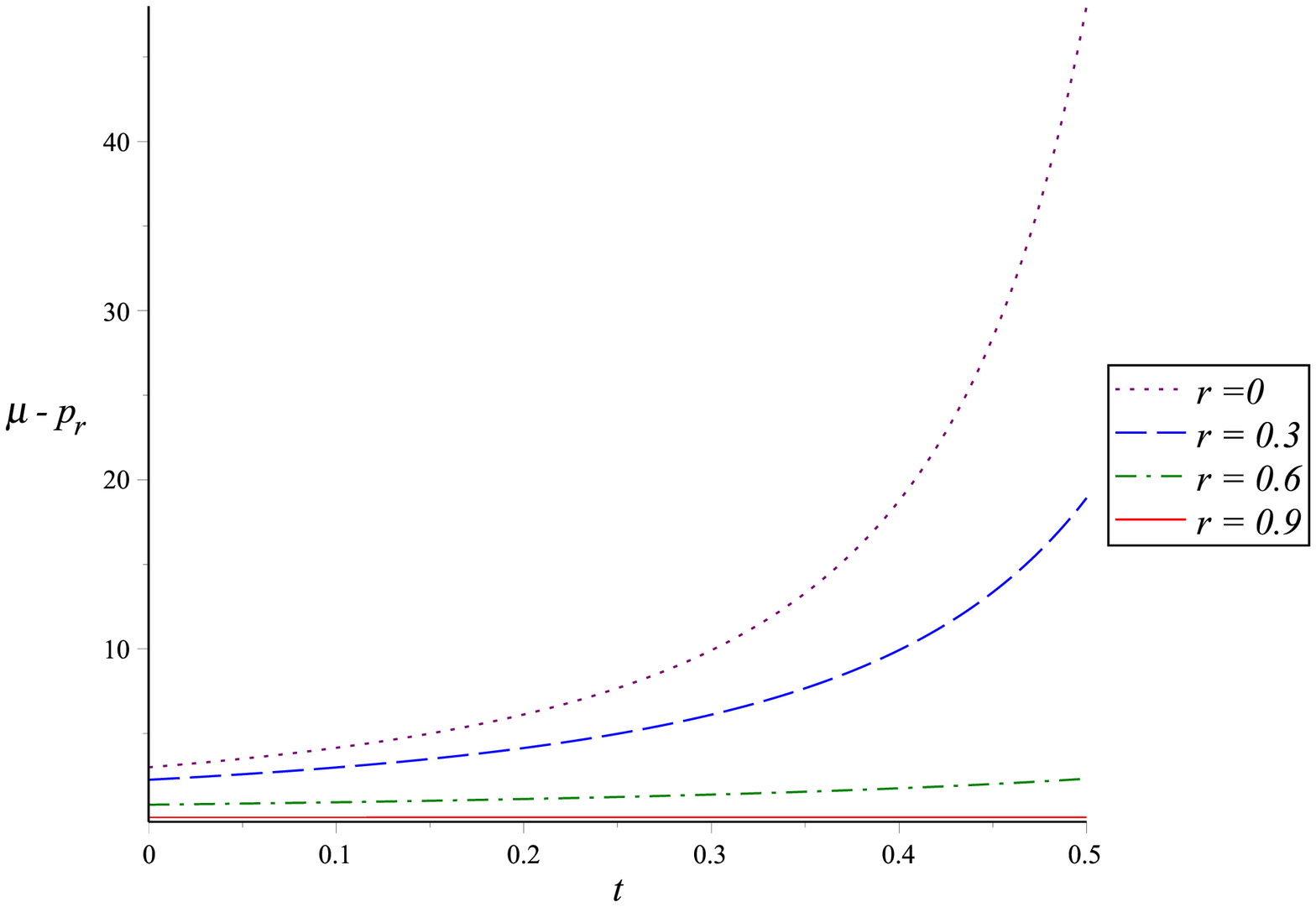}
}\\
\subfloat[DECII  \label{DECIIfig}]{%
\includegraphics[scale=0.33]{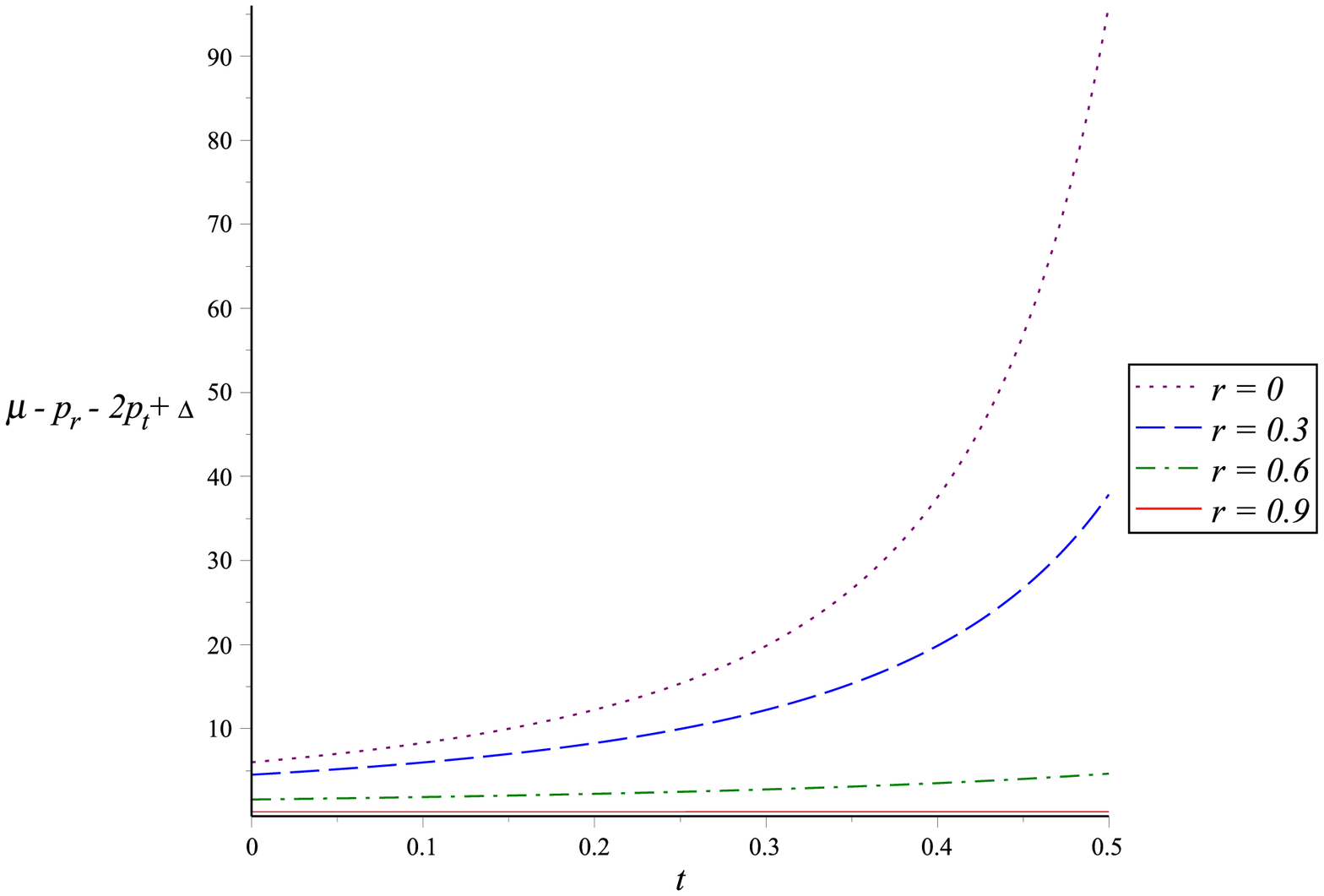}
}
\qquad \quad
\subfloat[SEC  \label{SECfig}]{%
\includegraphics[scale=0.33]{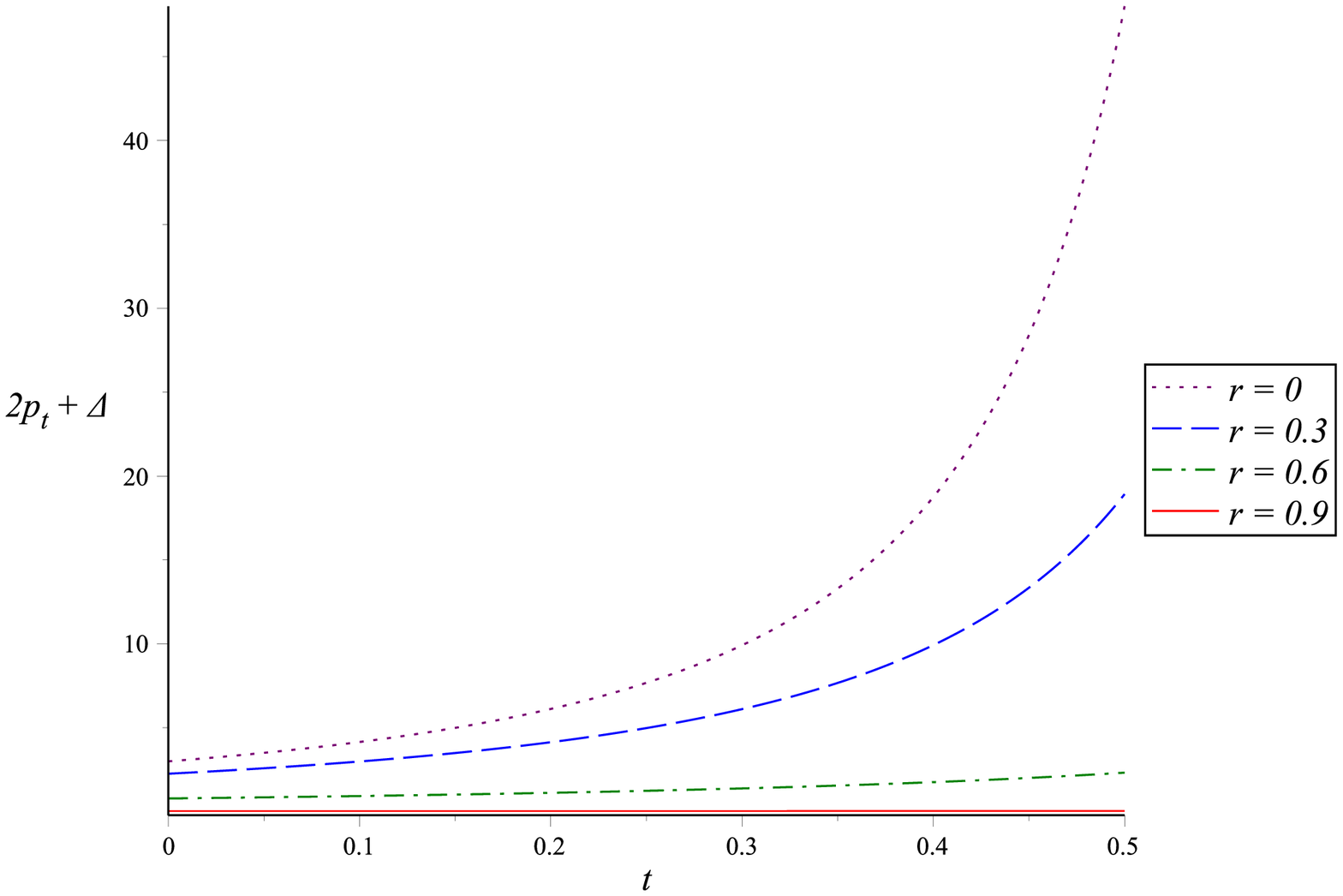}
}
\caption{Energy conditions \reff{EC1}-\reff{SEC} at different times.}
\label{ECfig}
\end{figure*}

\section{Concluding remarks}

In this paper we studied the possibility of the existence of spherically symmetric and astrophysically realistic collapsing stellar solutions in $f(R)$-gravity. The key results which emerged from our analysis are as follows:
\begin{enumerate}[(a)]
\item We showed transparently that the extra matching conditions in $f(R)$-gravity impose strong constraints on the stellar structure and thermodynamic properties which, in our opinion, are unphysical. These constraints make classes of physically realistic collapse scenarios in GR non-admissible in these theories. We showed that apart from a few special types of collapsing matter, these theories demand the collapsing matter to be inhomogeneous in order to smoothly match the interior spacetime with a static vacuum exterior as demanded by the astrophysical observations. Therefore, contrary to the belief that a higher-order theory will expand a set of admissible solutions, we find that the extra matching conditions are in fact contracting the set of physical models.
\item To show that the set of physically realistic collapsing solutions in $f(R)$-gravity is non-empty we explicitly found an analytic solution of a collapsing star with anisotropic pressure and heat flux in the interior for the Starobinski model. The matter in the interior of the star obeys all physically reasonable energy conditions and the interior of the star can be smoothly matched to a Schwarzschild exterior at the boundary. However, we demonstrated that the extra matching conditions in these theories strongly restrict the otherwise free functions of integration in the system. Hence we may conclude that these solutions are unstable to any matter perturbation in the stellar interior and consequently cannot describe a stable astrophysical collapse scenario.
\item It is interesting to note that our claim on the fine tuning of stellar thermodynamics in $f(R)$-gravity is supported by \cite{Vilija}. This paper investigates stellar structures for a broader class of scalar tensor theories. The authors found that the presence of a global potential for the scalar field (which in the context of our paper corresponds to a given theory of gravity), heavily constrains the allowed matter configuration of the star.
\item The Openheimer-Snyder-Datt model, a widely accepted collapse model for black hole formation via dynamical collapse, is no longer a viable model in the modified theories. Thus to establish the existence of a black hole via stellar collapse, we need to find new physically reasonable solutions. This may not be a simple task. For example, in the inhomogeneous solution which we found in this paper, it can be easily shown that the apparent horizon, which is the boundary of the trapped region in the spacetime, does not form early enough to shield the central singularity at $r=0$ from the external observers \cite{gos} and will produce a naked singularity as the end state of gravitational collapse. It is quite interesting that investigating the cosmic censorship hypothesis \cite{rp} (which states that naked singularities are not possible in a physically realistic collapse scenario) in modified theories of gravity may be far more difficult than in GR as it is well stablished that inhomogeneity is closely related to spacetime shear and the Weyl curvature of the collapsing matter \cite{Hamid,Joshi} which are the geometrical factors that produce a naked singularity.\\
\end{enumerate}
 
\begin{center}
{\bf Acknowledgments}
\end{center}
AMN and RG are supported by the National Research Foundation (South Africa). SDM acknowledges that this work is based upon 
research supported by the South African Research Chair Initiative of the Department of Science and Technology. SGG thanks the University of KwaZulu-Natal for financial support.


\end{document}